\documentclass[12pt]{article}

\usepackage{color} %% I need the stuff for a fancy Xfig figures

\usepackage{amsmath,amssymb,latexsym,graphicx}%
\usepackage[amsmath,thmmarks]{ntheorem}
\theoremseparator{.}%

\usepackage{hyperref}%
\hypersetup{%
   breaklinks,%
   colorlinks=true,%
   urlcolor=[rgb]{0.2,0.0,0.0},
   linkcolor=[rgb]{0.5,0.0,0.0},%
   citecolor=[rgb]{0,0,0.445},
   filecolor=[rgb]{0,0,0.4},
   anchorcolor=[rgb]={0.0,0.1,0.2}
}

\usepackage{styles/salgorithm}%

\usepackage{euscript}
\usepackage{graphicx}%
\graphicspath{{figs/}}

\usepackage[cm]{fullpage}

\usepackage{xspace}
\theoremstyle{theorem}%
\newtheorem{theorem}{Theorem}[section]
\newtheorem{lemma}[theorem]{Lemma}%
\newtheorem{observation}[theorem]{Observation} %

   \theoremheaderfont{\em}%
   \theorembodyfont{\upshape}%
   \theoremstyle{nonumberplain}
   \theoremseparator{}
   \theoremsymbol{\myqedsymbol}
   \newtheorem{proof}{Proof:}

   \newcommand{\myqedsymbol}{\rule{2mm}{2mm}}

\newcommand{\mmx}[1] {\ifmmode{#1}\else{\mbox{\(#1\)}}\fi}

\newcommand{\etal}{et al.\xspace}

\newcommand{\Rspace} {\mmx{{\mathbb R}}}  
\newcommand{\dist} {\mmx{\rm d}}  

\newcommand{\Frechet} {Fr\'{e}chet\xspace}%
\newcommand{\FDist} {\mmx {\rm \delta_{F}}}%
\newcommand{\DFDist} {\mmx {\rm \delta_{D}}}

\newcommand{\permut}[1]{\left\langle {#1} \right\rangle}%

\newcommand{\HLinkShort}[2]{\hyperref[#2]{#1\ref*{#2}}}
\newcommand{\HLink}[2]{\hyperref[#2]{#1~\ref*{#2}}}
\newcommand{\HLinkPage}[2]{\hyperref[#2]{#1~\ref*{#2}%
      $_\text{p\pageref{#2}}$}}
\newcommand{\HLinkPageOnly}[1]{\hyperref[#1]{Page~\refpage*{#1}%
      $_\text{p\pageref{#1}}$}}

\newcommand{\HLinkSuffix}[3]{\hyperref[#2]{#1\ref*{#2}{#3}}}
\newcommand{\HLinkPageSuffix}[3]{\hyperref[#2]{#1\ref*{#2}%
      #3$_\text{p\pageref{#2}}$}}

\newcommand{\figlab}[1]{\label{fig:#1}}
\newcommand{\figref}[1]{\HLink{Figure}{fig:#1}}

\newcommand{\seclab}[1]{\label{sec:#1}}
\newcommand{\secref}[1]{\HLink{Section}{sec:#1}}

\newcommand{\apndlab}[1]{\label{apnd:#1}}
\newcommand{\apndref}[1]{\HLink{Appendix}{apnd:#1}}

\newcommand{\lemlab}[1]{\label{lemma:#1}}
\newcommand{\lemref}[1]{\HLink{Lemma}{lemma:#1}}%

\newcommand{\obslab}[1]{\label{observation:#1}}
\newcommand{\obsref}[1]{\HLink{Observation}{observation:#1}}

\newcommand{\thmlab}[1]{{\label{theo:#1}}}
\newcommand{\thmref}[1]{\HLink{Theorem}{theo:#1}}

\providecommand{\eqlab}[1]{}%
\renewcommand{\eqlab}[1]{\label{equation:#1}}

\newcommand{\backbone}{backbone\xspace} \newcommand{\Pone} {{\sf P1}}
\newcommand{\Ptwo} {{\sf P2}} \newcommand{\ball} {B}
\newcommand{\error} {\delta} \newcommand{\Amap} {D}
 \newcommand{\Amapsimp}{\EuScript{D}}
\newcommand{\newp}{\hat{p}} \newcommand{\newq} {\hat{q}}
\newcommand{\pnew} {\hat{p}} \newcommand{\qnew} {\hat{q}}
 
\newcommand{\algD} {{\sf ApprDecision}}  
\newcommand{\apprFbackbone} {{\sf ApprFBackbone}}
\newcommand{\smallexactalg} {{\sf ExactFSmall}} \newcommand{\yes}
{\mmx {\rm yes}} \newcommand{\no} {\mmx {\rm no}}

\newcommand{\matching}{correspondence\xspace} \newcommand{\length}{l}

\newcommand{\rank} {\ensuremath{\text{\textsc{rank}}_d}}

\newcommand{\bwnumber}{number of \switchingcells}   
 \newcommand{\bwcells} {\S}
\newcommand{\switchingcells} {switching cells}
\newcommand{\Switchingcells} {Switching cells}
\newcommand{\switchingcell} {switching cell}
\DeclareMathOperator{\polylog}{polylog} \newcommand{\column} {C}
\newcommand{\reach} {V}

\newcommand{\A}{\mathcal{A}}%

\newcommand{\remove}[1]{}

\newcommand{\mergeCol} {{\sf mergeColumn}} \newcommand{\newreach} {X}
\newcommand{\potentialReach} {\mmx {\rm potentialReachFlag}}
\newcommand{\wantReach} {\mmx {\rm potentialStartFlag}}
\newcommand{\EC}{\EuScript{C}}

\newcommand{\eps}{{\varepsilon}}%%

\newcommand{\curve}{\mathbf{\tau}} \newcommand{\curveA}{\mathbf{\pi}}
\newcommand{\curveB}{\mathbf{\sigma}}
\newcommand{\Asimp}{\widetilde{\curveA}}
\newcommand{\Bsimp}{\widetilde{\curveB}}

\renewcommand{\Re}{\mathbb{R}}%
\newcommand{\reals}{\Re}%

\newcommand{\Interval}{\EuScript{I}}

\newcommand{\TFuzzDecProc}[2]{T_{\textsc{FDecision}}(#1, #2)}

\renewcommand{\th}{th\xspace}
\providecommand{\si}[1]{#1}

\IfFileExists{sariel_computer.sty}{\def\sarielComp{1}}{}
\ifx\sarielComp\undefined%
\newcommand{\SarielComp}[1]{}
\newcommand{\NotSarielComp}[1]{#1} \else
\newcommand{\SarielComp}[1]{#1}%
\newcommand{\NotSarielComp}[1]{}%
\fi

\newcommand{\email}[1]{\href{mailto:#1}{#1}}%

\begin{document}
\title{\Frechet{} Distance for Curves, Revisited%
   \footnote{A preliminary version of this paper appeared in ESA 2006
      \cite{ahkww-fdcr-06}.}%
}

%% BA removed tt fonts.  Look too ugly for my taste.
\author{%
   Boris Aronov%
   \thanks{%
      Dept. of Comp. Sci. \& Engineering; %
      Polytechnic School of Engineering; %
      New York University, NY; %
      Research supported in part by NSF ITR Grant CCR-00-81964 and by
      a grant from US-Israel Binational Science Foundation. %
      \url{http://cis.poly.edu/\string~aronov}.%
   }%
   \and%
   Sariel Har-Peled%
   \thanks{Dept. of Comp. Sci, University of Illinois; 1304 West
      Springfield Ave., Urbana, IL 61801;
      \email{sariel@uiuc.edu}.
      \url{http://sarielhp.org}.%
   }%
   \and%
   Christian Knauer%
   \thanks{%
      Universit\"at Bayreuth; %
      Institut f\"ur Angewandte Informatik; %
      95440 Bayreuth, Germany; %
      \email{christian.knauer@uni-bayreuth.de}.  }%
   \and%
   Yusu Wang\thanks{%
      Dept. of Comp. Sci. and Engineering, The Ohio State Univ,
      Columbus, OH 43016; \email{yusu@cse.ohio-state.edu}.
      \url{http://www.cse.ohio-state.edu/\string~yusu/}.%
   }%
   \and%
   Carola Wenk%
   \thanks{%
      Department of Computer Science, %
      Tulane University, %
      New Orleans, LA 70118, %
      \email{cwenk@tulane.edu}, %
      \url{http://www.cs.tulane.edu/\string~carola}.%
   }%
}

% \date{Processed by \LaTeX\ \today\\
% \small\texttt{\timestamp}} \date{}

\date{\today}

\maketitle
% \vspace*{-0.25in}
\begin{abstract}
    We revisit the problem of computing \Frechet{} distance between
    polygonal curves under $L_1$, $L_2$, and $L_\infty$ norms,
    focusing on discrete \Frechet{} distance, where only distance
    between vertices is considered.  We develop efficient algorithms
    for two natural classes of curves. In particular, given two
    polygonal curves of $n$ vertices each, a $\eps$-approximation of
    their discrete \Frechet{} distance can be computed in roughly
    $O(n\kappa^3\log n/\eps^3)$ time in three dimensions, if one of
    the curves is \emph{$\kappa$-bounded}. Previously, only a
    $\kappa$-approximation algorithm was known. If both curves are the
    so-called \emph{\backbone~curves}, which are widely used to model
    protein backbones in molecular biology, we can $\eps$-approximate
    their \Frechet{} distance in near linear time in two dimensions,
    and in roughly $O(n^{4/3}\log nm)$ time in three dimensions.  In
    the second part, we propose a pseudo--output-sensitive algorithm
    for computing \Frechet{} distance exactly. The complexity of the
    algorithm is a function of a quantity we call the
    \emph{\bwnumber{}}, which is quadratic in the worst case, but
    tends to be much smaller in practice.
\end{abstract}

\section{Introduction}
\seclab{intro}

\Frechet{} metric is a natural measure of similarity between two
curves \cite{eghmm-nsmpa-02}.  An intuitive definition of the \Frechet
distance is to imagine that a dog and its handler are walking on their
respective curves.  Both can control their speed but can only go
forward.  The \Frechet{} distance of these two curves is the minimal
length of any leash necessary for the handler and the dog to move from
the starting points of the two curves to their respective endpoints.
\Frechet{} distance and its variants have been widely used in many
applications such as in dynamic time-warping \cite{kp-sudtw-99},
speech recognition \cite{khmtc-pgbhp-98}, signature verification
\cite{pp-carcd-90}, and matching of time series in databases
\cite{kks-osmut-05}.

Alt \etal\cite{ag-cfdbt-95} present an algorithm to compute the
\Frechet distance between two polygonal curves of $n$ and $m$
vertices, respectively, in time $O(nm\log^2 (nm))$.  Improving this
roughly quadratic-time solution for general curves seems to be hard,
and so far, no algorithm, exact or approximate, with running time
significantly smaller than $O(nm)$ has been found for this problem for
general curves.%
% \footnote{We remark that the running time can be improved to
% $O(nm/\log nm)$ by using the technique from \cite{CLZ02}.}
Since the \Frechet{} distance essentially requires computing a
correspondence between the two curves, it has some resemblance to the
edit distance problem (which asks for the best alignment of two
strings), for which no substantially subquadratic algorithm is known
either.

On the other hand, another similarity measure, the \emph{Hausdorff
   distance},
% is better understood, and
can be computed faster in the plane and approximated efficiently in
higher dimensions.  Unfortunately, Hausdorff distance does not reflect
curve similarity well (see \figref{fig:kbounded} (a) for an example).
Alt \etal \cite{akw-cdmpc-04} showed that the Hausdorff distance and the
\Frechet{} distance are the same for a pair of closed convex curves.
They also showed that the two measures are closely related for
$\kappa$-bounded curves. Roughly speaking, for any two points $p, q$
on a $\kappa$-bounded curve $\curve$, $\curve(p,q)$, the subcurve from
$p$ to $q$, is contained within some neighborhood of $p$ and $q$ with
size roughly $\kappa || p-q ||$ (see \figref{fig:kbounded} (b);
precise definition is introduced later).  Alt~et~al.\ showed that the
\Frechet{} distance between any two $\kappa$-bounded curves is bounded
by $\kappa+1$ times the Hausdorff distance between them. This leads to
a $\kappa$-approximation algorithm for the \Frechet{} distance for any
pair of $\kappa$-bounded curves, and they also developed an algorithm
to compute the reparametrizations of input curves that realize this
approximation (see the definitions below).  The algorithm runs is
$O((n+m)\log^2 (n+m)2^{\alpha(n+m)})$ time in two dimensions.  In
three or higher dimensions, the time complexity is dominated by the
computation of Hausdorff distance between the curves.  Not much is
known about the \Frechet{} distance for other types of curves. In
fact, even for $x$-monotone curves in three dimensions, no known
algorithm runs in substantially subquadratic time.

The problem of minimizing \Frechet{} distance under various classes of
transformations has also been studied \cite{akw-mpcrf-01, w-smhd-02}.
However, even in two dimensions, the exact algorithm takes roughly
$O(n^6)$ time for computing the best \Frechet{} distance under
translations, and roughly $O(n^8)$ time under rigid motions.
Approximation algorithms have been studied \cite{cm-ampcr-05,
   w-smhd-02}, but practical solutions remain elusive. The basic
building block of those algorithms, as well as one of the bottlenecks,
is the computation (or approximation) of the \Frechet{} distance
between curves $\curveA$ and $\curveB$.

There is a slightly simpler version of the \Frechet~distance, the
\emph{discrete \Frechet~distance}, which only consider vertices of
polygonal curves. Its computation takes $\Theta(n^2)$ time and space
using dynamic programming \cite{em-cdfd-94}, and no substantially
subquadratic algorithm is known either.  \Frechet{} distance has also
been extended to graphs (maps) \cite{aerw-mpm-03}, to piecewise smooth
curves \cite{r-cfdps-05}, to simple polygons \cite{bbw-cfdsp-06} and
to surfaces \cite{ab-scfds-05}.  Finally, \Frechet{} distance was used
as the similarity measure for morphing \cite{eghm-mbp-01} between
curves, and for high-dimensional approximate nearest neighbor search
\cite{i-annaf-02}. It was also used for efficient curve simplification
\cite{ahmw-nltaa-05}.
\begin{figure}[htb]
    \begin{center}
        \begin{tabular}{ccccc}
          \includegraphics[width=4cm]{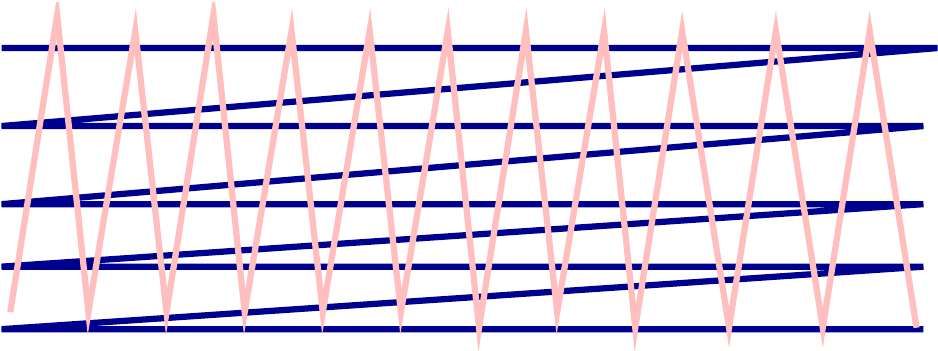}%
          &%
            \hspace*{0.1in}%
          &%
            \includegraphics{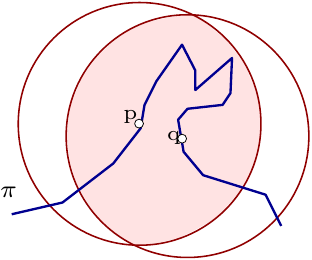}%
          &%
          % &
              \hspace*{0.1in}%
          &%
            \includegraphics[width=5cm]{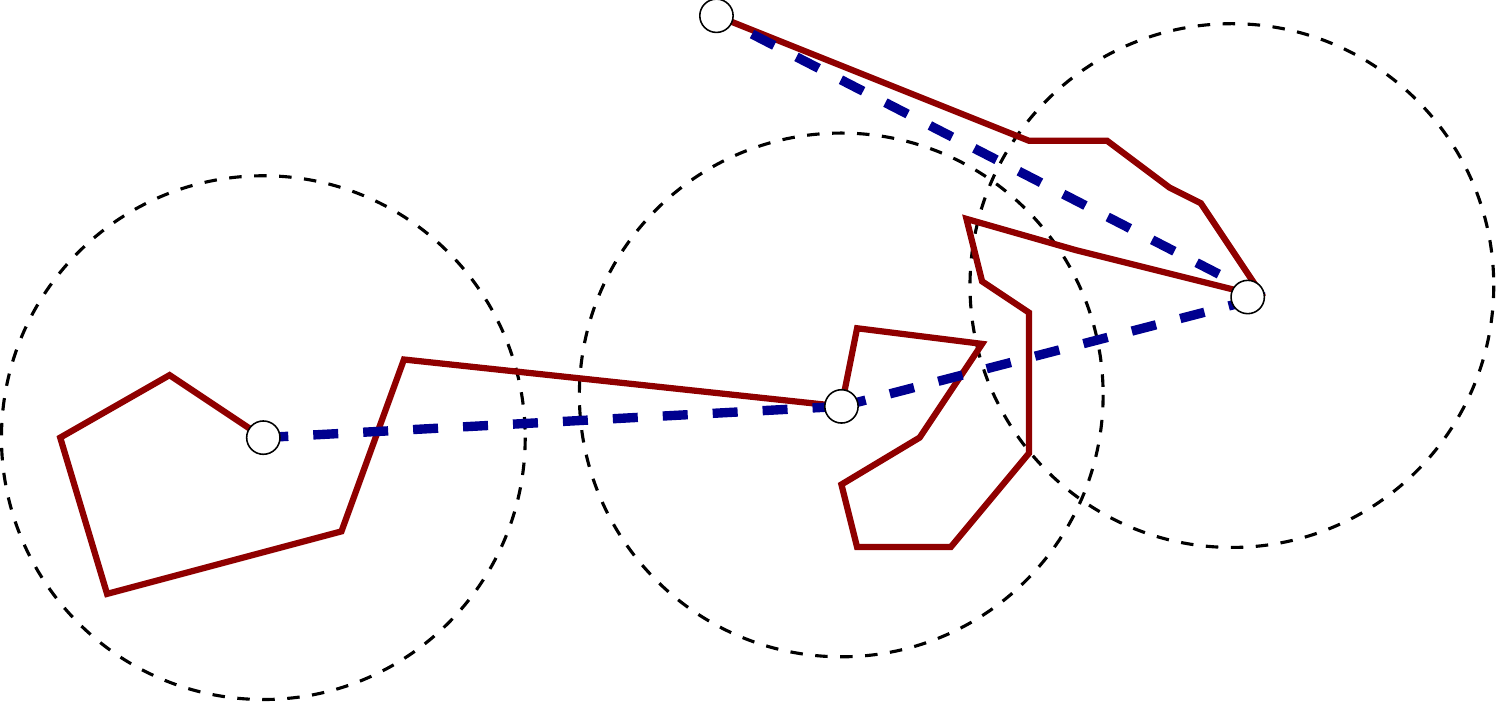} \\
          (a) & & (b) & & (c)
        \end{tabular}
    \end{center}
    \caption{(a) Light and dark curves are close under Hausdorff but
       far under \Frechet~distance.  (b) $\curveA$ is
       \emph{$\kappa$-bounded} iff for any $p, q \in \curveA$,
       subchain $\curveA(p,q)$ lies inside the shaded region, where
       the radius of two disks (centered at $p$ and $q$) is
       $\kappa\dist(p,q)/2$.  (b) The dashed curve $\mu$-simplifies
       the solid one; the radius of each disk is $\mu$.}
    \figlab{fig:kbounded} \vspace*{-0.2in}
\end{figure}
\paragraph{Our results. }
Given the apparent difficulty of improving the worst-case time
complexity of computing the \Frechet{} distance between two
unrestricted polygonal curves, we aim at developing algorithms for
more realistic cases. First, in \secref{approximation:algorithms}, we
consider efficient approximation algorithms for the slightly simpler
variant of \Frechet{} distance, the \emph{discrete \Frechet{}
   distance}, the best algorithms for which currently have only
slightly better worst-case time complexity than the continuous case.
Most currently algorithms for computing \Frechet distance rely on a
so-called \emph{decision} procedure which determines whether a given
distance is larger or smaller than the \Frechet distance between the
two given curves.  We observe that an approximation solution to the
decision problem can lead to an approximation of \Frechet{} distance,
and curve simplification can help us to approximate the decision
problem efficiently. We apply this idea for two families of common
curves. In the first case, given two polygonal curves of size $n$ and
$m$ respectively, with one of them being $\kappa$-bounded, we can
$\eps$-approximate their discrete \Frechet{} distance in
$O((m + n \kappa^d/\eps^d) \log (nm))$ time in $d$-dimensions.  In the
second case, both curves are so-called \emph{\backbone~curves}, used
widely to model molecular structures like protein backbones, and
DNA/\si{RNAs}. We $\eps$-approximate their (both \emph{discrete} and
\emph{continuous}) \Frechet{} distance in near linear time in two
dimensions, and in roughly $O(nm^{1/3}\log nm)$ time in three
dimensions.

In \secref{b:w:cells}, we shift our focus back to the exact
computation of discrete \Frechet{} distance. Previously, the problem
of deciding whether \Frechet{} distance was smaller than some
threshold was cast as finding some viable path in the so-called
\emph{free-space diagram} which is a $n \times m$ map. We observe that
such viable path can be computed once some subset $\bwcells$ of cells
in the free-space diagram are given. The size of $\bwcells$ is $nm$ in
worst case, but is expected to be much smaller in general. Based on
this observation, we present algorithms that run in
$O(|\bwcells| + n\log^{d-1} n)$ time for discrete \Frechet{} distance
in $L_\infty$ norm in $d$-dimensions.  For $L_2$ norm, it takes
roughly $O(|\bwcells| + n^{4/3}\polylog n)$ time in two dimensions,
and $O(|\bwcells| + n^{2-1/2^d}\polylog n)$ time for $d > 2$.

\section{Preliminaries} 
\seclab{preliminaries}

A \emph{(parameterized) curve} in $\Rspace^d$ can be represented as a
function $f \colon [0,1] \to \Rspace^d$.  A \emph{(monotone)
   reparametrization} $\alpha$ is a continuous non-decreasing function
$\alpha \colon [0,1] \to [0,1]$ with $\alpha(0)=0$ and $\alpha(1)=1$.
Given two curves $f,g : [0,1] \to \Rspace^d$, the \emph{\Frechet
   distance} between them, $\FDist(f, g)$, is defined as
\[
    \FDist(f,g) := \inf_{\alpha,\beta} \max_{t \in [0, 1]}
    \dist(f(\alpha(t)), g(\beta(t)) ).
\]
where $\dist(x, y)$ denotes the Euclidean distance between points $x$
and $y$, and $\alpha$ and $\beta$ range over all monotone
reparametrizations.

\paragraph{Discrete \Frechet{} Distance.}
A simpler variant of the \Frechet{} distance for two polygonal curves
$\curveA = \permut{ p_1, p_2, \ldots, p_n }$ and
$\curveB = \permut{ q_1, q_2, \ldots, q_m }$ is the \emph{discrete
   \Frechet{} distance}, denoted by $\DFDist(\curveA,
\curveB)$.
Imagine that both the dog and its handler can only stop at vertices of
$\curveA$ and $\curveB$, and at any step, each of them can either stay
at their current vertex or jump to the next one (i.e., magically, both
the dog and the handler seem to have turned into frog princess and
prince, respectively).  The discrete \Frechet{} distance is defined as
the minimal leash necessary at these discrete moments.

To formally define the discrete \Frechet{} distance, we first consider
a discrete analog of continuous reparametrizations.  A \emph{discrete
   monotone reparametrization} $\alpha$ from $\{ 1, \ldots, k \}$ to
$\{ 1, \ldots, \ell \}$ is a non-decreasing function
$\alpha : \{ 1, \ldots, k \} \to \{ 1, \ldots, \ell \}$, for integers
$k \geq \ell \geq 1$, with $\alpha(1)=1, \alpha(k)=\ell$ and
$\alpha(i+1)\leq\alpha(i)+1$, for all $i=1,\ldots,k-1$.  An
\emph{(order-preserving complete) \matching} between $\curveA$ and
$\curveB$ is a pair $(\alpha,\beta)$ of discrete monotone
reparametrizations from $\{1, \ldots, k \}$ to $\{1, \ldots, m \}$ and
$\{1, \ldots, n \}$.
% , where $m$ and $n$ are the numbers of vertices of $\curveA$ and
% $\curveB$, respectively.
The \emph{discrete \Frechet{} distance} between $\curveA$ and
$\curveB$, $\DFDist( \curveA, \curveB )$ is
\[
    \DFDist(f,g) := \min_{(\alpha,\beta)} \max_{t \in [1, k]}
    \dist(f(\alpha(t)), g(\beta(t)) ),
\]
where $(\alpha,\beta)$ range over all order-preserving complete
\matching{}s between $\curveA$ and $\curveB$. An equivalent definition
of order-preserving complete \matching~between $\curveA$ and $\curveB$
is a set of pairs
$M \subseteq \{ (p, q) \mid p \in \curveA, q \in \curveB \}$ such that
(i) \emph{order-preserving}: if $(p_i, q_j) \in M$, then no
$(p_s, q_t) \in M$ for $s < i$ and $t > j$; and (ii) \emph{complete:}
for any $p \in \curveA$ (resp. $q \in \curveB$), there exists some
pair involving $p$ (resp.  $q$) in $M$.  The discrete \Frechet{}
distance is related to the edit distance between the ``strings''
$\curveA$ and $\curveB$ where the cost of changing a symbol is the
Euclidean distance of the relevant points.

It is well known that discrete and continuous versions of the \Frechet
distance relate to each other as follows:
\[
    \FDist(\curveA, \curveB) \le \DFDist(\curveA, \curveB) \le
    \FDist(\curveA, \curveB) + \max \{\ell_1, \ell_2\},
\]
where $\ell_1$ and $\ell_2$ are the lengths of the longest edges in
$\curveA$ and $\curveB$, respectively.  This suggests using $\DFDist$
to approximate $\FDist$.
% Discrete \Frechet{} distance can be computed in $\Theta(nm)$ time
% using dynamic programming directly \cite{em-cdfd-94}, which removes
% the parametric search as involved in the continuous case.
Unfortunately, it seems that computing $\DFDist(\curveA, \curveB)$ is
asymptotically almost as hard as computing $\FDist(\curveA, \curveB)$.

\paragraph{Decision problem.}
In the original paper, Alt and Godau \cite{ag-cfdbt-95} used the
following framework to compute $\FDist(\curveA, \curveB)$: First,
develop a procedure that answers the following \emph{decision problem}
in $\Theta(nm)$ time and space by a dynamic programming algorithm:
Given a parameter $\error \ge 0$, is
$\FDist(\curveA, \curveB) \le \error$?  This procedure is then used as
a subroutine to search for $\FDist(\curveA, \curveB)$ using parametric
search paradigm within $O(nm\log^2 nm)$ time\cite{ag-cfdbt-95,
   ast-apsgo-94}. The same paradigm can be used to compute
$\DFDist(\curveA, \curveB)$ in $O(nm\log (nm))$ time by replacing the
parametric search to a binary search. Although this is slightly worse
than the $\Theta(nm)$ algorithm in \cite{em-cdfd-94}, we describe how
to solve the decision problem for $\DFDist(\curveA, \curveB)$ below,
as our algorithm will use this framework, and as the algorithm from
\cite{em-cdfd-94} runs in $\Theta(nm)$ time for any input.
\begin{figure}[htb]
    \centering \vspace{0.1in}
    % \begin{tabular}{c@{\hspace*{0.1\columnwidth}}c}
    \begin{tabular}{cccc}
      \includegraphics{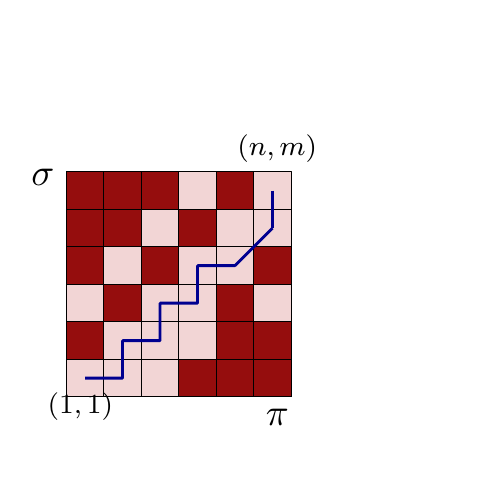}%
      &%
        \includegraphics{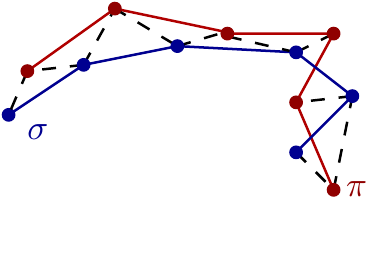}%
      & 
        \includegraphics{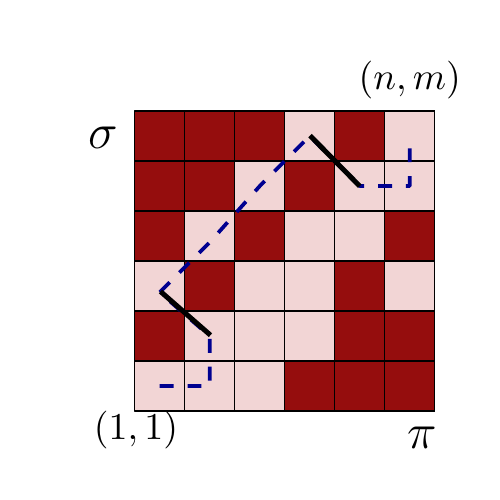}%
      &%
        \includegraphics{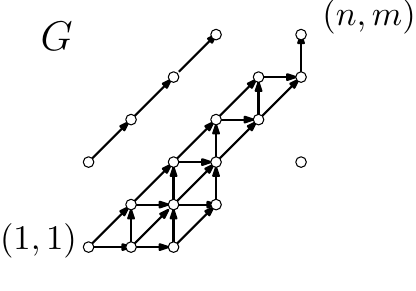} \\
      (a) & (b) & (c) & (d)\\
      % \includegraphics[width=0.25\columnwidth]{matrix_c.eps}&
      % \includegraphics[width=0.3\columnwidth]{matrix_d.eps}\\
      % (c) & (d)
    \end{tabular}
    \caption{The valid path (solid curve) in the Free-space diagram
       $D(\curveA, \curveB, \error)$ in (a) corresponds to the
       order-preserving, complete \matching (dashed lines) in (b). The
       path in (c) is not valid, as the two solid segments violate the
       bi-monotonicity condition.  (d) The directed graph
       corresponding to the white cells in free-space diagram in (a).}
    \figlab{DPexplanation}
\end{figure}

Given two polygonal chains $\curveA$ and $\curveB$ and a distance
threshold $\error \ge 0$, we construct the following \emph{free-space
   diagram} $\Amap=\Amap( \curveA, \curveB, \error)$; $\Amap$ is an
$n \times m$ matrix (grid) and a grid cell $\Amap[i,j]$ has value $1$
if $\dist(p_i, q_j) \le \error$, and value $0$ otherwise.  We refer to
$1$-cells as \emph{white} and $0$-cells as \emph{black}.
% See \figref{DPexplanation}(a) for an illustration.
The white cells in the $i$th column (resp., $j$th row) correspond to
the set of vertices of $\curveB$ (resp., $\curveA$) whose distance to
$p_i$ (resp., $q_j$) is less than $\error$.  A \emph{viable} path in
$\Amap$ is a path connecting $s := \Amap[1,1]$ to $t := \Amap[n,m]$,
visiting only white cells of $\Amap$, and moving in one step from
$(i,j)$ to either $(i,j+1)$, $(i+1,j)$, or $(i+1,j+1)$.
% ; we call any path that satisfies the last condition
% \emph{bi-monotone}.
It is easy to check that a complete order-preserving \matching{} $M$
induces a viable path in $\Amap$ and vice versa (see
\figref{DPexplanation}). Hence
% conversely the existence of a viable path implies the existence of a
% complete order-preserving \matching{} between $\curveA$ and
% $\curveB$ (see \figref{DPexplanation}). Hence
the problem of deciding ``$\DFDist(\curveA, \curveB) \leq \error$?''
is equivalent to deciding the existence of a viable path in $\Amap$.

Given $\Amap$, one can extract a viable path, if it exists, in
$\Theta(nm)$ time by a dynamic programming algorithm.  Alternatively,
one can traverse a directed graph $G$ defined as follows: The nodes of
$G$ are the white cells of $\Amap$.  A white cell is connected to its
top, right, or top-right neighbor cells by a directed edge, if they
are white. See \figref{DPexplanation} (d). The size of $G$ is the
bounded by $|W|$, the number of white cells of $\Amap$, the in-degree
(out-degree) of each node is at most three, and
$\DFDist(\curveA, \curveB) \le \error$ if and only if there is a
directed path in $G$ from $(1,1)$ to $(n,m)$.  That is, testing this
condition corresponds to a connectivity check in a directed graph in
time $O(|W|)$, once the graph $G$ is given.

\paragraph{Approximations.}
We say that $\tau$ is an $\eps$-approximation of
$\delta(\curveA, \curveB)$ if
\[
    (1 - \eps) \delta(\curveA, \curveB) \le \tau \le (1+\eps)
    \delta(\curveA, \curveB).
\]
We say that an algorithm \emph{$\eps$-approximates} the decision
problem ``Is $\delta(\curveA, \curveB) \le \delta$?'', if it returns
`yes' whenever $\delta(\curveA, \curveB) \le (1-\eps) \delta$ and `no'
whenever $\delta(\curveA, \curveB) \ge (1+\eps) \delta$.  If $\delta$
is a $(1+\eps)$-approximation of $\delta(\curveA,\curveB)$, the
algorithm is allowed to return either `yes' or `no.' Such an algorithm
is also called \emph{an $\eps$-fuzzy} decision procedure for
$\delta(\curveA, \curveB)$.

\section{Approximation Algorithms Based on Simplification}
\seclab{approximation:algorithms}

In this section, we first introduce a general framework for
approximating the discrete \Frechet{} distance by solving the decision
problem approximately. We then present efficient approximation
algorithms for two families of common curves based on this framework:
the $\kappa$-bounded curves and the \backbone curves, using curve
simplifications, packing arguments, and other observations.

\subsection{Approximation via approximate decision problem}
\seclab{apprdecision}

Given a set $P$ of $N$ points in $\Re^d$, compute a
\emph{well-separated pairs decomposition (WSPD)} of $P$ for a
separation parameter $10$, which is a collection $\{(A_i,B_i)\}$ of
pairs of subsets of $P$, with the property that (1) for every pair of
points $x,y \in P$, there is an index $i$, so that $x \in A_i$ and
$y \in B_i$ and (2) the minimum distance between $A_i$ and $B_i$ is at
least 10 times the diameter of either set. One can compute such a
collection of size $O(N)$ in $O(N \log N)$ time \cite{ck-dmpsa-95}.
For every pair $(A_i,B_i)$ in the WSPD, we choose an arbitrary pair of
points $p_i \in \A_i$ and $q_i \in B_i$ as its representatives. It is
easy to check that the distance between any two points $x,y \in P$ is
$1/5$-approximated by the distance between the representatives of the
corresponding WSPD pair.

If we want to approximately solve an optimization problem using a
decision procedure, where the optimal solution $\delta^*$ is one of
the distances induced by a pair of points of $P$, then we can use the
above WSPD to extract $O(N)$ values: for each WSPD pair, we take the
distance between its representative points.  Next, we replace each
value $x$ by two values $\frac45x$ and $\frac65x$, sort the resulting
values, and perform a binary search (using the decision procedure) to
identify which interval delimited by consecutive values contains
$\delta^*$. Let $\Interval = [x,y]$ be the resulting interval;
obviously $y \leq \frac65x$. We now perform another binary search on
this interval to identify the interval $[x',y']$ containing $\delta^*$
with $y' \leq (1+\eps)x'$, giving rise to an $\eps$-approximation of
$\delta^*$. The second binary search invokes the decision procedure
$O( \log(1 /\eps) )$ times.

Interestingly, the decision procedure does not have to be exact, and
it can return a \emph{fuzzy} answer, in the sense of last section. An
equivalent view of an $\eps$-fuzzy decision procedure is: for a
parameter $\delta$, if it returns ``no'', then
$\delta^* < (1+\eps)\delta$; otherwise if it returns ``yes'', then
$\delta^* > (1-\eps)\delta$. It can be shown that the above binary
search can be adapted to work with a fuzzy decision procedure with the
same performance guarantees (details omitted and can be found in
\apndref{A}. We summarize:
\begin{theorem}
    \thmlab{fuzzy}%
    Let $P$ be a set of $N$ points in $\Re^d$, and let $X$ be an
    optimization problem, for which the optimal answer is a distance
    induced by a pair of points of $P$.  Given an $\eps$-fuzzy
    decision procedure for $X$, one can $\eps$-approximate the optimal
    solution in
    \[
        O(N \log N + \TFuzzDecProc{N}{1/10} \log N +
        \TFuzzDecProc{N}{\eps/4} \log (1/\eps))
    \]
    time, where $\TFuzzDecProc{N}{\epsilon}$ is running time of the
    fuzzy decision procedure when the required accuracy is $\epsilon$.
\end{theorem}

\begin{proof}
    The algorithm is described above.  The fuzzy decision procedure
    can be used with constant accuracy in the stage of the
    algorithm. Higher accuracy of $\eps/4$ is required only at the
    second stage, when we perform the binary search over the interval
    $\Interval$.
\end{proof}

On the other hand, observe that there must exist some
$p^* \in \curveA$ and $q^* \in \curveB$ such that
$\dist(p^*, q^*) = \DFDist(\curveA, \curveB)$. In other words, the
solution $\delta^* = \DFDist(\curveA, \curveB)$ will be one of the
distances induced by a pair of points from $P = \{$ vertices from
$\curveA$ and $\curveB \}$. Hence the above theorem implies that we
now only need a fuzzy decision procedure for
$\DFDist(\curveA, \curveB)$ in order to approximate
$\DFDist(\curveA, \curveB)$.

\subsection{Approximation with simplifications}

The remaining question is how to implement fuzzy decision procedure
efficiently. One useful heuristic is curve simplification.  Below we
first describe the particular simplification we use and how it helps
in approximating $\DFDist(\curveA, \curveB)$. We then show that
together with a packing argument and other observations, guaranteed
efficiency can be achieved for the two classes of common curves that
we investigate.

\paragraph{Greedy simplification.}
Given a polygonal chain $\curveA = \permut{ p_1, \ldots, p_n }$, we
simplify $\curveA$ to obtain
$\Asimp = \permut{ \pnew_1, \ldots, \pnew_k}$, where vertices of
$\Asimp$ form a subsequence of $\curveA$, with $\pnew_1 = p_1$ and
$\pnew_k = p_n$.  More precisely, let $I_\curveA(i) = j$ if
$\pnew_i = p_j \in \curveA$; the subscript $\curveA$ is omitted when
it is clear from context.  We say that $\Asimp$
\emph{$\mu$-simplifies} $\curveA$ if (i) $I(i) < I(k)$ for $i < k$
(i.e, order-preserving), and (ii) $\dist(\pnew_i, p_k) \le \mu$ for
any $k \in [~I(i), I(i+1)~)$ (see \figref{fig:kbounded} (c)).  (This
definition of $\mu$-simplification is slightly different from the
standard definition found in the literature.)

We construct a $\mu$-simplification of $\curveA$, $\Asimp$, in a
greedy manner: Start with $\pnew_1 = p_1$. At some stage, suppose we
have already computed $\pnew_i = p_j$.  In order to find $I(i+1)$, we
check each vertex of $\curveA$ starting from $p_j$ in order, and stop
when we reach the first edge $p_kp_{k+1}$ of $\curveA$ such that
$\dist(p_j,p_k) \leq \mu$ and $\dist(p_j,p_{k+1}) > \mu$. We set
$\pnew_{i+1} = p_{k+1}$ and proceed until we reach $p_n$, at which
point we add $p_n$ as the last vertex of $\Asimp$. The entire
procedure takes linear time. By construction, the following
observation is straightforward.

\begin{observation}
    For any edge $\newp_i\newp_{i+1}$ in $\Asimp$, other than the last
    edge,
    %% Overkill, in my opinion
    % i.e, $i \in [1, |\Asimp|-1 )$,
    we have $\dist(\newp_i, \newp_{i+1}) \ge \mu$.

    \obslab{simplength}
\end{observation}
\begin{figure}[htb]
    \begin{center}
        \includegraphics{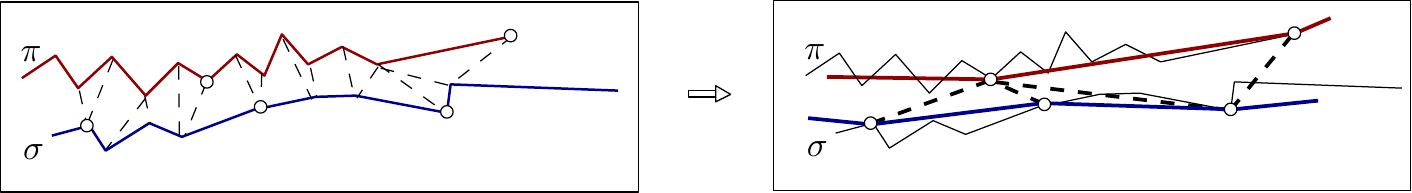}
        \caption{Small empty circles mark vertices of $\Asimp$ and
           $\Bsimp$. Left picture shows part of $\curveA$ and
           $\curveB$ and the correspondence $M^*$ (indicated by dashed
           segments). Right picture shows $\Asimp$ and $\Bsimp$ (thick
           curves) and the induced correspondence for them (thick
           dashed segments). }
        \figlab{fig:simpappr}
    \end{center}
    \vspace*{-0.2in}
\end{figure}
Now if we $\mu$-simplify both input curves $\curveA$ and $\curveB$ to
obtain $\Asimp$ and $\Bsimp$, we have the following lemma:
\begin{lemma} %LEMMA C
    \lemlab{simpappr}%
    $\DFDist(\curveA, \curveB) - 2\mu \le \DFDist(\Asimp, \Bsimp) \le
    \DFDist(\curveA, \curveB) + \mu$.
\end{lemma}
\begin{proof}
    We first consider the right-hand inequality.  Let
    $\error^* = \DFDist(\curveA, \curveB)$, and $M^*$ be the complete
    order-preserving \matching{} that produces
    $\DFDist(\curveA, \curveB)$.  Obviously, for any pair $(p_i, q_j)$
    in $M^*$, we have that $\dist(p_i, q_j) \le \error^*$. $M^*$ can
    be modified into a \matching{} $M$ between $\Asimp$ and $\Bsimp$
    as follows: we add the match $(\newp_i, \newq_j)$ to $M$ if and
    only if (C1) there exists a match
    $(p_{I_\curveA(i)}, q_b) \in M^*$ such that
    $b \in [I_\curveB(j), I_\curveB(j+1) )$, or (C2) there exist a
    match $(p_a, q_{I_\curveB(j)}) \in M^*$ such that
    $a \in [I_\curveA(i), I_\curveA(i+1) )$, where $I_\curveA$
    (resp. $I_\curveB$) maps the indices between $\curveA$ and
    $\Asimp$ (resp. $\curveB$ and $\Bsimp$) (see \figref{fig:simpappr}
    for an example).  It is easy to verify that $M$ is both complete
    and order-preserving. By the triangle inequality, we have that
    $\dist(\newp_i,\newq_j) \le \dist(\newp_i, q_b) + \dist(q_b,
    \newq_j)$
    for case (C1), implying that
    $\dist(\newp_i, \newq_j) \le \dist(\newp_i, q_b) + \mu \le
    \error^* + \mu$
    (the case for (C2) is symmetric).  Since $\DFDist(\Asimp, \Bsimp)$
    should be smaller than the distance induced by $M$, the right-hand
    inequality then follows.

    The proof for the left-hand inequality is similar but slightly
    more involved. Details omitted and can be found in the
    \apndref{B}.
\end{proof}

The above lemma implies that if the answer to
$\DFDist(\Asimp, \Bsimp) \le \error$ is `yes', then,
$\DFDist(\curveA, \curveB) \le \error + 2\mu$.  If it is `no', then
$\DFDist(\curveA, \curveB) \ge \error - \mu$.  Thus the decision
problem for $\DFDist(\Asimp, \Bsimp)$ ($2\mu/\error$)-approximates
that of $\DFDist(\curveA, \curveB)$. We next show that
$\DFDist(\Asimp, \Bsimp)$ can be answered asymptotically much faster
for two special classes of curves, giving rise to efficient fuzzy
decision procedure for them.

\subsection{\Frechet~Distance for $\kappa$-bounded curves}
\seclab{k-bounded}

Given a polygonal curve $\curveA$, let
$\curveA(x, y) \subseteq \curveA$ denote the subcurve of $\curveA$
that connects $x \in \curveA$ and $y \in \curveA$, and
$\length_{\curveA}(x,y)$ the length of $\curveA(x,y)$ along $\curveA$;
$\curveA$ may be omitted from the subscript when clear. We say that
$\curveA$ is \emph{$\kappa$-straight} if
$\length(x,y) \le \kappa \cdot \dist(x, y) $, for any
$x, y \in \curveA$.  Examples of $\kappa$-straight curves include
\emph{curves with increasing chords} of \cite{r-cic-94} and
\emph{self-approaching curves} of \cite{aaikl-gsac-01}.  As defined
by Alt \etal~\cite{akw-cdmpc-04}, $\curveA$ is \emph{$\kappa$-bounded}
if
$\curveA(x, y) \subseteq \ball(x, \frac{\kappa}{2} \dist(x,y) ) \cup
\ball(y, \frac{\kappa}{2} \dist(x, y))$,
for all $x,y \in \curveA$, where $\ball(x, r)$ is the radius-$r$
Euclidean ball centered at $x$ and where we have slightly abused the
notation by treating a curve section as a point set. See
\figref{fig:kbounded} (b) for an illustration in two dimensions.
Every $\kappa$-straight curve is $\kappa$-bounded.

We now describe how to construct an $\eps$-fuzzy decision procedure
for the problem ``$\DFDist(\curveA, \curveB) \le \error$?'', where one
curve, say $\curveB$, is $\kappa$-bounded.  We first $\mu$-simplify
$\curveA$ and $\curveB$ into $\Asimp$ and $\Bsimp$ respectively, using
$\mu := \eps \error / 2$. By \lemref{simpappr}, the decision problem
for $\DFDist(\Asimp, \Bsimp)$ is an $\eps$-fuzzy decision procedure
for $\DFDist(\curveA, \curveB)$. Hence we now focus on checking
whether $\DFDist(\Asimp, \Bsimp) \le \error$. Let $n$, $m$, $r$, $s$
be the size of $\curveA$, $\curveB$, $\Asimp$, and $\Bsimp$
respectively; $r = O(n)$ and $s = O(m)$.

\paragraph{Decision problem for $\DFDist(\Asimp, \Bsimp)$. }
Let $\Amapsimp$ be the free-space diagram for $\Asimp$ and $\Bsimp$
with respect to $\delta$. Recall that
$\DFDist(\Asimp, \Bsimp) \le \error$ if there exists a viable path in
$\Amapsimp$ which can be computed in $O(|W|)$ time once $W$, the set
of white cells of $\Amapsimp$ are given. We first bound the size of
$W$.

For every $\newp \in \Asimp$, let $N(\newp)$ be the set of points from
$\Bsimp$ contained in $\ball(\newp, \error)$. Obviously,
$|W| = \sum_{\newp \in \Asimp} | N(\newp) |$. Consider any two points
$q_1, q_2 \in \Bsimp$ that lie in $\ball(\newp, \error)$ for some
$\newp \in \Asimp$. There are two cases: (i) $q_1q_2$ is an edge of
$\Bsimp$ and (ii) otherwise.  For case (i), we have that
$\dist(q_1, q_2) \ge \mu$ by \obsref{simplength}. For case (ii), we
know that
\[
    \curveB(q_1, q_2) \subseteq \ball(q_1, \frac{\kappa}{2} \dist(q_1,
    q_2)) \cup \ball(q_2, \frac{\kappa}{2} \dist(q_1, q_2)),
\]
as $\curveB$ is $\kappa$-bounded. Furthermore, let
$q_1q \subset \Bsimp$ be the edge with $q \in \curveB(q_1, q_2)$;
$\dist(q_1, q) \ge \mu$ by \obsref{simplength}. It then follows that
$(\kappa/2)\dist(q_1, q_2) \ge \mu$ and therefore
$\dist(q_1, q_2) \ge 2\mu / \kappa$.  Hence
$N(\newp) = O((\kappa \error / \mu)^d)$ by a straightforward packing
argument. This means that the number of white cells is
$|W| = O(s (\kappa \error/\mu)^d) = O(n (\kappa\error/\mu)^d)$ given
that $\curveB$ is a $\kappa$-bounded curve.

We still need to compute $N(\pnew)$ efficiently, that is, to enumerate
the set of vertices of $\Bsimp$ contained in $\ball(\pnew, \error)$
for every $\pnew \in \Asimp$. This can be done by a spherical range
query. As there are no known efficient algorithms for spherical range
queries, we instead first perform a $\beta$-approximate range query of
$\ball(\pnew, \error)$ among all vertices from $\Bsimp$, such that
vertices lying completely inside $\ball(\pnew, \error)$ are guaranteed
to be retrieved, those completely outside
$\ball(\pnew, (1+\beta)\error)$ will not be reported, while those
in-between may or may not be returned. By the same packing argument as
above, it is easy to verify that the number of vertices returned is
still bounded by $O((\frac{(1+\beta)\kappa\error}{\mu})^d)$. We then
inspect each vertex returned, and only mark the corresponding cell in
$\Amapsimp$ white when it indeed lies in $\ball(\pnew, \error)$.

We preprocess $\Bsimp$ into a data structure of size $O(s) = O(m)$,
using $O(s)$ preprocessing time, such that the resulting
data-structure answers $\beta$-approximate range query for
$\ball(\pnew, \error)$ in $O(1/\beta^d)$ time. This can be easily
achieved by constructing a grid of appropriate size (which is
$\beta \error$), throwing the points of $\Bsimp$ into this grid (using
hashing). Next, an approximate spherical range query is no more than
probing all the grid cells that intersects the query ball. The number
of cells being probed in a single query is $O(1/\beta^d)$.  Therefore
the set of white cells in $\Amapsimp$ can be computed in
$O(r+s+ r (\kappa\error/\mu)^d )$ time by choosing $\beta > 0$ to be a
small constant, say $\beta = 1/2$.

Putting everything together, we have an $\eps$-fuzzy decision
procedure for $\DFDist(\curveA, \curveB)$ that runs in
$O(n+m + n \kappa^d/\eps^d)$ time and space in $\reals^d$. By
\thmref{fuzzy}, we have that:
\begin{lemma}
    An $\eps$-approximation of $\DFDist(\curveA, \curveB)$ for a
    polygonal curve $\curveA$ and a $\kappa$-bounded curve $\curveB$,
    of size $n$ and $m$ respectively, can be computed in
    $O((m + n \kappa^d/\eps^d) \log (n/\eps))$ time and
    $O(n+m+n \kappa^d / \eps^d)$ space in $d$~dimensions.
\end{lemma}

\subsection{\Frechet{} Distance for Protein Backbones}
\seclab{backbone}

In molecular biology, it is common to model a protein backbone by a
polygonal chain, where each $C_{\alpha}$ atom becomes a vertex, and
each edge represents a covalent bond between two consequent amino
acids. All the bonds have approximately the same bond length, and no
two atoms (thus vertices) can get too close due to van der Waals
interactions.  This is the motivation behind the study of the
\emph{\backbone{}} curves, which have the following properties:
\begin{itemize}
    \item[P1.] For any two non-consecutive vertices $u$ and $v$ of the
    curve, $\dist(u, v) \ge 1$,

    \item[P2.] Every edge of the curve has length $l$ such that
    $c_1 \le l \le c_2$, where $c_1, c_2 > 0$ are constants.
\end{itemize}
We remark that although proteins lie in three dimensional space, there
are simplified models for protein backbones in both two and three
dimensions, such as the lattice model which has been widely studied to
understand the mechanism behind protein folding \cite{gip-aapss-99,
   ks-mcspf-94}.

Now suppose we are given \backbone curves $\curveA$ and $\curveB$ in
$\reals^d$.  Given a distance threshold $\error \ge 0$, we want to
know whether $\DFDist(\curveA, \curveB) \le \error$. We $\mu$-simplify
$\curveA$ and $\curveB$ to obtain $\Asimp$ and $\Bsimp$ as in the
previous case, for $\mu = \eps \error/ 2$, and construct the
free-space diagram $\Amapsimp$ for $\Asimp$ and $\Bsimp$ with respect
to $\error$.  $\Amapsimp$ is an $r \times s$~grid, where by
\obsref{simplength} and property \Ptwo, $r = |\Asimp| \le c_2 n / \mu$
and $s = |\Bsimp| \le c_2 m / \mu$.  Once $\Amapsimp$ is given, the
decision problem can be solved in time proportional to $|W|$, where
$W$ is the set of white cells in $\Amapsimp$.

\paragraph{The set of white cells $W$. }
A straightforward bound for $|W|$ is
$O(\min \{r \error^d, s \error^d\})$\footnote{In the following, the
   big $O$ notation sometimes hide factors depending on constants
   $c_1$ and $c_2$. }, as by the packing argument and property \Pone,
there are at most $O(\error^d)$ vertices lying in
$\error$-neighborhood of any vertex of $\Asimp$ and $\Bsimp$. If
$\delta < 1$, then the number of white cells is $O(n+m)$. Hence we now
assume that $\delta \ge 1$.

We can improve this bound by a more careful counting analysis. Assume
without loss of generality that $r \le s$.  For any vertex
$\pnew \in \Asimp$ and its $\error$-neighborhood
$\ball(\pnew, \error)$, let $E(\pnew)$ be the set of edges of $\Bsimp$
intersecting the ball $\ball(\pnew, \error)$.  The number of vertices
of $\Bsimp$ in $\ball(\pnew, \error)$ can be upper bounded by
$O(|E(\pnew)|)$.  Furthermore, given any edge
$e = (\qnew_i, \qnew_{i+1}) \in \Bsimp$, let
$\curveB(e) = \curveB(q_{I_q(i)}, q_{I_q(i+1)})$ (that is, subchain
$\curveB(e) \subseteq \curveB$ is simplified into edge $e$ in chain
$\Bsimp$).  $E(\pnew)$ can be partitioned into two sets: (i)
$E_1 = \{ e \in E(\pnew) \mid \curveB(e) \subseteq \ball(\pnew,
\error) \}$,
and (ii)
$E_2 = \{ e \in E(\pnew) \mid \text{at least a vertex of}~ \curveB(e)
~\text{lies outside~} \ball(\pnew, \error) \}$.
 
By property \Ptwo, we know that the number of vertices in $\curveB(e)$
is at least $\mu / c_2$ for any $e \in \Bsimp$.  Therefore
$|E_1| = O(c_2 \error^d / \mu)$. On the other hand, for every edge
$e \in E_2$, there is at least one vertex of $\curveB(e)$ that lies in
the spherical shell of
$\ball(\pnew, \error + c_2) \setminus \ball(\pnew, \error)$, as the
length of edges in $\curveB$ is at most $c_2$. Since the volume of
this spherical shell is $O(c_2(c_2 + \error)^{d-1})$, the size of
$E_2$ is bounded by $O((c_2 (c_2+\error)^{d-1}/(c_1^{d-1})))$.
Therefore, we have that
$|E(\pnew)| = |E_1| + |E_2| = O(\error^{d-1} + \error^d / \mu)$.
Summing it over all $r$ vertices of $\Asimp$, we have that
$|W| = O(\frac{n}{\mu} (\error^{d-1} + \error^d / \mu)) $.
Furthermore, since this number cannot exceed the size of $\Amapsimp$
which is $O(rs) = O(nm/\mu^2)$, we have
$|W| = \min \{ nm/\mu^2, O(\frac{n}{\mu} (\error^{d-1} + \error^d /
\mu) \}$.
Note that $|W|$ is maximized when the two balancing terms are equal:
$\frac{nm}{\eps^2\error^2} = \frac{\error^{d-2}}{\eps^2}$, that is,
when $\error = m^{1/d}$.  This implies that
$|W| = O(n m^{1-2/d} / \eps^2)$.

We still need to compute these white cells of $\Amapsimp$
efficiently. Similar to the case for $\kappa$-bounded curves, we
preprocess $\Bsimp$ into a data structure of size $O(s)$, using $O(s)$
preprocessing time, such that the resulting data structure answers
$\beta$-approximate range query for $\ball(\pnew, \error)$ in
$O(1/\beta^d)$ time, for a small constant $\beta > 0$, say
$\beta = 1/2$. We then check all vertices returned by this approximate
range query, and keep only those indeed contained in
$\ball(\pnew, \error)$.  Overall, we can compute all white cells in
$O(r+s+|W|)$ time and space, thus can answer the decision problem ``Is
$\DFDist(\Asimp, \Bsimp) \le \error ?$ '' in the same time and space.
Putting everything together, we have:
\begin{lemma}
    \lemlab{simpdecision}%
    Given two \backbone~curves of sizes $n$ and $m$, respectively, we
    can develop an $\eps$-fuzzy decision procedure for
    $\DFDist(\curveA, \curveB)$ w.r.t. $\delta$ that runs in
    $O((n+m) + \frac{1}{\eps^2}n m^{1 - 2/d})$ time and space. In
    particular, the time complexity is $O(n+m/\eps^2)$ when $d=2$, and
    $O(n+m + nm^{1/3}/\eps^2)$ when $d = 3$.    
\end{lemma}

Finally, for \backbone~curves, in order to approximate
$\DFDist(\curveA, \curveB)$, one can use a binary search procedure
(described in \apndref{C} instead of the approach using WSPD as
described earlier. The advantage of the binary search procedure is
that all results can then be extended for the continuous case
$\FDist(\curveA, \curveB)$ by more careful and involved packing
arguments. We conclude with the following theorem.
\begin{theorem}
    \thmlab{backboneappr}%
    Given two \backbone curves $\curveA$ and $\curveB$ of $n$ and $m$
    vertices respectively, we can compute an $\eps$-approximation of
    $\FDist(\curveA, \curveB)$ in $O(\frac{(n+m)}{\eps^3}\log (nm))$
    time in two dimensions, and $O(\frac{1}{\eps^3}nm^{1/3}\log (nm))$
    time in three dimensions. 
\end{theorem}

% ----------------------------------------------------------------------
% ----------------------------------------------------------------------
\section{Pseudo--Output-Sensitive Algorithm}
\seclab{b:w:cells} In this section, given curves $\curveA$ and
$\curveB$ of size $n$ and $m$, respectively, we present a
pseudo-output-sensitive algorithm for computing
$\DFDist(\curveA, \curveB)$ for general curves. Although the worst
case complexity may still be $\Theta(nm)$, we believe that the
observation made within should help to produce efficient (possibly
approximate) algorithms for \Frechet~distance in practice. In what
follows, we provide results for $L_\infty$ norm (which provides a
constant factor approximation for optimal solution under $L_2$
norm). The time complexity for exact computation under $L_2$ norm is
quite messy and omitted.

Suppose we have an algorithm that answer the following
\emph{select-distance} query in $B(N)$ time: given a set of $N$ points
$P$ and a rank $k$, what is $\rank(k)$, the $k$\th smallest distance
among all pair-wise distances from $P$. Now given an algorithm to
solve the decision problem ``Is
$\delta^* = \DFDist(\curveA, \curveB) \le \delta$?'' in time $A(n+m)$,
we can find the optimal solution $\delta^*$ in
$O((A(n+m) + B(n+m)) \log (nm))$ time by querying $\rank(k)$ among
$n+m$ points in a binary search manner\footnote{In some sense, our
   previous approach using WSPD is performing implicit approximate
   distance selection. }. For $L_\infty$ norm, the distance-selection
problem can be solved in $O(dN\log^{d-1} N)$ in $\reals^d$
\cite{s-lisps-89}.
% For $L_2$ norm, it can be solved in $O(N^{4/3+\nu})$ in $\reals^2$
% and $O(N^{2-1/2^d+\nu})$ in $\reals^d$ for $d>2$, where $\nu > 0$ is
% an arbitrarily small constant. For $L_1$ norm, it takes $O(N\log N)$
% time in $\reals^2$ and same time as the case for $L_2$ norm in
% higher dimensions\footnote{Approximate solution for the distance
% selection problem can be found in \cite{}. }.
For the decision problem, a straightforward bound for time complexity
$A$ is $|W|$ plus the time to compute $W$, where $W$ is the set of
white cells in the free-space diagram
$\Amap = \Amap(\curveA, \curveB, \error)$ for a threshold
$\error > 0$. Below we provide a tighter bound for $A$ although its
worst-case complexity is still $\Theta(nm)$.

\paragraph{\Switchingcells. }
Given an $n \times m$~map $\Amap$ with respect to some threshold
$\error$, a \emph{\switchingcell}~ is a white cell whose immediate
neighbor above or below it is black. So if $\Amap[i,j]$ is a
\switchingcell, then the edge $q_jq_{j+1} \subset \curveB$ (or
$q_j q_{j-1}$) intersects the boundary of $\ball(p_i, \error)$ exactly
once (one endpoint must lie inside and one must be outside). For a
vertex $p\in \curveA$, while the set of white cells involving $p$
correspond those vertices from $q$ falling inside $\ball(p, \error)$,
the \switchingcells~involving $p$ correspond to those vertices inside
$\ball(p, \error)$ with one incident edge crossing the boundary of
$\ball(p, \error)$.  Let
$\bwcells = \bwcells(\curveA, \curveB, \error)$ denote the set of
\switchingcells{} of $\Amap(\curveA, \curveB, \error)$.  Although in
worst case $|\bwcells| = \Omega(|W|) = \Omega(nm)$, we expect it to be
much smaller than $|W|$ in practice. For example, consider the case
when vertices of $\curveB$ form lines of a cubic lattice of size
$n^{1/3} \times n^{1/3} \times n^{1/3}$ and $\error$ is roughly
$n^{1/3}/2$. For a vertex $p$ at the center of this cube, the number
of white cells in the corresponding column in $\Amap$ is $\Theta(n)$,
while the number of \switchingcells~is $\Theta(n^{2/3})$.  The
remaining questions are (i) how to compute the set of
\switchingcells~$\bwcells(\curveA, \curveB, \error)$ and (ii) how to
solve the decision problem once $\bwcells$ is given.

\paragraph{Decision problem with $\bwcells$. }
Once the set of \switchingcells{} is given, we can solve the decision
problem in $O(|\bwcells|)$ time and space as follows. Instead of
representing $\Amap$ explicitly, we now represent each column of
$\Amap$, $\column[i]$ for $1 \le i \le n$, as a set of ordered
intervals, where each interval corresponds to a maximal set of
consecutive white cells in this column. Obviously, the endpoints of
these intervals are exactly the \switchingcells. Let $\reach[i]$ be
the set of ordered intervals, each representing a maximal set of cells
in the $i$\th column reachable from $\Amap[1,1]$, and $|\column[i]|$
and $|\reach[i]|$ the number of intervals in $\column[i]$ and
$\reach[i]$, respectively. Easy to see that
$|\reach[i]| \le |\column[i]|$, because all cells covered by intervals
from $\reach[i]$ are white, and because if any cell $c$ from an
interval $I \in \column[i]$ is in some interval $J \in \column[i]$,
then all cells of $I$ above $c$ should also be covered by $J$.  Our
algorithm scans $\Amap$ from left to right (i.e, from column $1$ to
column $n$), and at the $i$\th round, we compute $\reach[i]$ by
merging $\column[i]$ and $\reach[i-1]$ in
$O(|\reach[i-1]| + |\column[i]|)$ time using a merge-sort like
procedure (see details in \apndref{D}).

\paragraph{Computing $\bwcells$. }
Given $p$ and $\error$, let $\bwcells(p, \delta)$ denote the set of
edges from $\curveB$ ``crossing'' the boundary of $\ball(p, \error)$.
Here by \emph{crossing}, we mean that one endpoint of the edge is
inside $\ball(p, \error)$ and one is outside (so it is not the usual
segment/ball intersection problem).  To compute $\bwcells$, we need to
perform $n$ \emph{edge/ball crossing queries}, one for each vertex
from $\curveA$.  Under the $L_\infty$ norm, the basic operation is in
fact an edge/cube crossing query, where all cubes are congruent. We
can preprocess the set of edges by building a range-search tree for
their endpoints (similar to the multi-level data structure for
orthogonal range reporting problem). The entire data structure has
size $O(m\log^{2d} m)$ and given a cube, the set of edges crossing it
can be reported in $O(\log^{2d} m + k)$ where $k$ is the number of
such edges.

Putting everything together, we conclude with the following theorem:
\begin{theorem}
    Given two arbitrary polygonal curves $\curveA$ and $\curveB$ in
    $\reals^d$, with $n$ and $m$ vertices, respectively, one can
    compute $\DFDist(\curveA, \curveB)$ under $L_\infty$-norm, in
    $O((\Phi + (n+m)\log^{2d} (nm) ) \log (nm) )$ time and
    $O(\Phi + (n+m)\log^{2d} (nm))$ space,
        %             % \item[(ii)]
        %         and (ii) under $L_2$-norm, in
        %         %         $\tilde{O}((\Phi + (n+m)^{4/3})\log (nm))$
        %         %         time for
        %         %         $d=2$, and in
        %         %         $\tilde{O}((\Phi + (n+m)^{2-1/2^d}) \log
        %         %         (nm) )$
        %         %         time for $d > 2$.
        %     \end{itemize}
    where $\Phi$ is an upper bound of the number of
    \switchingcells~for any threshold $\error$.
\end{theorem}

We remark that for $L_2$ norm, the running time is
$\tilde{O}(\Phi + (n+m)^{4/3} \log (nm))$ for $d=2$ and
$\tilde{O}((\Phi + (n+m)^{2-1/2^d}) \log (nm) )$ time for $d > 2$. The
edge/ball crossing query required by computing $\bwcells$ can be
converted into an segment/hyperplane query in one dimension higher. It
is less practical as the solution involves heavy
machinery. Nevertheless, if approximation is allowed, one can use the
idea from \thmref{fuzzy} as well as multi-dimensional range trees to
obtain an $\eps$-approximation algorithm in $O(\Phi + \polylog~n)$
time and space where $\polylog$ depends on both $\eps$ and $d$.

\section{Conclusions and Discussion}
\seclab{discussion} In this paper, we considered the problem of
computing discrete \Frechet~distance between two polygonal curves
either approximately or exactly. Our main contribution is a simple
approximation framework that leads to efficient $\eps$-approximation
algorithms for two families of common curves: the $\kappa$-bounded
curves and the \backbone~curves. We also consider the exact algorithm
for general curves, and proposed a pseudo-output-sensitive algorithm
by observing that only a subset of the white cells from the free-space
diagram are necessary for the decision problem. It will be interesting
to investigate whether there are families of curves that are
guaranteed to have small $\Phi$, which is the upper bound on the
number of \switchingcells.

We feel that for general curves, it might be hard to develop
algorithms that are significantly sub-quadratic in worst case, given
that no such algorithm exists for a related and widely studied
problem, the edit distance for strings. Hence our future directions
will focus on practical variants of \Frechet~distance so that one can
handle outliers and/or partial matching, or so that one can perform
efficient multiple-curve alignments. Another important direction is to
develop efficient (approximation) algorithm for computing smallest
\Frechet~distance under rigid motions (in particular rotations).

\paragraph{Postscript.}

Since the appearance of this paper in ESA 2006 \cite{ahkww-fdcr-06} a
lot of research was done on related problems. Driemel \etal
\cite{dhw-afdrc-12} introduced the notion of $c$-packed curves, and
showed a near linear time algorithm for such curves.  Bringmann
\cite{b-wwdtt-14} proved that \Frechet distance can not be computed
exactly in subquadratic time under the SETH hypothesis. There is more
recent research on the \Frechet distance, but surveying it is outside
the scope of this note.

%*flatex input: [./discrete_frechet.bbl]
\newcommand{\etalchar}[1]{$^{#1}$}
 \providecommand{\CNFX}[1]{ {\em{\textrm{(#1)}}}}
  \providecommand{\tildegen}{{\protect\raisebox{-0.1cm}{\symbol{'176}\hspace{-0.03cm}}}}
  \providecommand{\SarielWWWPapersAddr}{http://sarielhp.org/p/}
  \providecommand{\SarielWWWPapers}{http://sarielhp.org/p/}
  \providecommand{\urlSarielPaper}[1]{\href{\SarielWWWPapersAddr/#1}{\SarielWWWPapers{}/#1}}
  \providecommand{\Badoiu}{B\u{a}doiu}
  \providecommand{\Barany}{B{\'a}r{\'a}ny}
  \providecommand{\Bronimman}{Br{\"o}nnimann}  \providecommand{\Erdos}{Erd{\H
  o}s}  \providecommand{\Gartner}{G{\"a}rtner}
  \providecommand{\Matousek}{Matou{\v s}ek}
  \providecommand{\Merigot}{M{\'{}e}rigot}
  \providecommand{\CNFSoCG}{\CNFX{SoCG}}
  \providecommand{\CNFCCCG}{\CNFX{CCCG}}
  \providecommand{\CNFFOCS}{\CNFX{FOCS}}
  \providecommand{\CNFSODA}{\CNFX{SODA}}
  \providecommand{\CNFSTOC}{\CNFX{STOC}}
  \providecommand{\CNFBROADNETS}{\CNFX{BROADNETS}}
  \providecommand{\CNFESA}{\CNFX{ESA}}
  \providecommand{\CNFFSTTCS}{\CNFX{FSTTCS}}
  \providecommand{\CNFIJCAI}{\CNFX{IJCAI}}
  \providecommand{\CNFINFOCOM}{\CNFX{INFOCOM}}
  \providecommand{\CNFIPCO}{\CNFX{IPCO}}
  \providecommand{\CNFISAAC}{\CNFX{ISAAC}}
  \providecommand{\CNFLICS}{\CNFX{LICS}}
  \providecommand{\CNFPODS}{\CNFX{PODS}}
  \providecommand{\CNFSWAT}{\CNFX{SWAT}}
  \providecommand{\CNFWADS}{\CNFX{WADS}}

% flatex input end: [./discrete_frechet.bbl]

%\NotSarielComp{\bibliographystyle{alpha}}

%FLATEX-REM:\bibliography{discrete_frechet} %
%\bibliography{shortcuts,geometry} %

% \newpage
\appendix
% \section{Appendix}
\section{Approximation via fuzzy decision procedure}
\apndlab{A}

Given an $\beta$-fuzzy decision procedure \algD($\delta$, $\beta$) for
deciding whether $\delta^* \le \delta$, we combine it with the WSPD
approach described in \secref{apprdecision} to compute an
$\eps$-approximation of $\delta^*$. In particular, construct the
$O(n)$ distances using WSPD as before, and perform a binary search by
querying \algD($v$, $1/10$) among these distances to identify the
interval $\Interval = [x,y]$ such that \algD($x, 1/10$) returns ``no''
while \algD($y,1/10$) returns ``yes''. Easy to verify that
$a = \frac{4}{5} x < \delta^* < \frac{7}{5} x = b$. We then start with
a pair $k_l=a / (1+\eps)$, and $k_h = b/(1-\eps)$, and perform a
standard binary search while always maintaining that
\algD($k_l, \eps/4$) returns ``no'' and \algD($k_h, \eps/4$) returns
``yes'' until $k_h - k_l \le (b-a) \eps / 3$. It is easy to verify
that the invariant holds when we start, and the number of iterations
is at most $O(\log (1/\eps) )$. Furthermore, because of the invariant
that we maintain, we have
$(1-\eps/4) k_l \le \delta^* \le (1+ \eps/4) k_h$ and
$k_h - k_l \le (b-a) \eps / 2 < \delta^* \eps / 2$. It then follows
that when $\eps < 1$,
$$\delta^* \le (1+\eps/4) k_h \le (1+\eps/4) (k_l + \delta^* \eps / 2) \rightarrow \delta^* \le (1+\eps/4) k_l / (1 - (1+\eps/4) \eps / 2) \le (1+\eps) k_l. $$
This implies that $k_l$ is an $\eps$-approximation of
$\delta^*$. Hence we can use a fuzzy decision procedure to approximate
$\delta^*$.

\section{Left-inequality of \lemref{simpappr} }
\apndlab{B}

Let $C^*$ be the complete order-preserving correspondence that produce
$\DFDist(\Asimp, \Bsimp)$. We now modify it into a correspondence $C$
between $\curveA$ and $\curveB$ as follows: First we add all matches
$(p_{I_\curveA(i)}, q_{I_\curveB(j)})$ to $C$ if
$(\pnew_i, \qnew_j) \in \EC^*$. Next, we take each pair of consecutive
matches. There are three cases as illustrated in
\figref{simpapprappendix}.
\begin{figure}[htb]
    \begin{center}
        \includegraphics{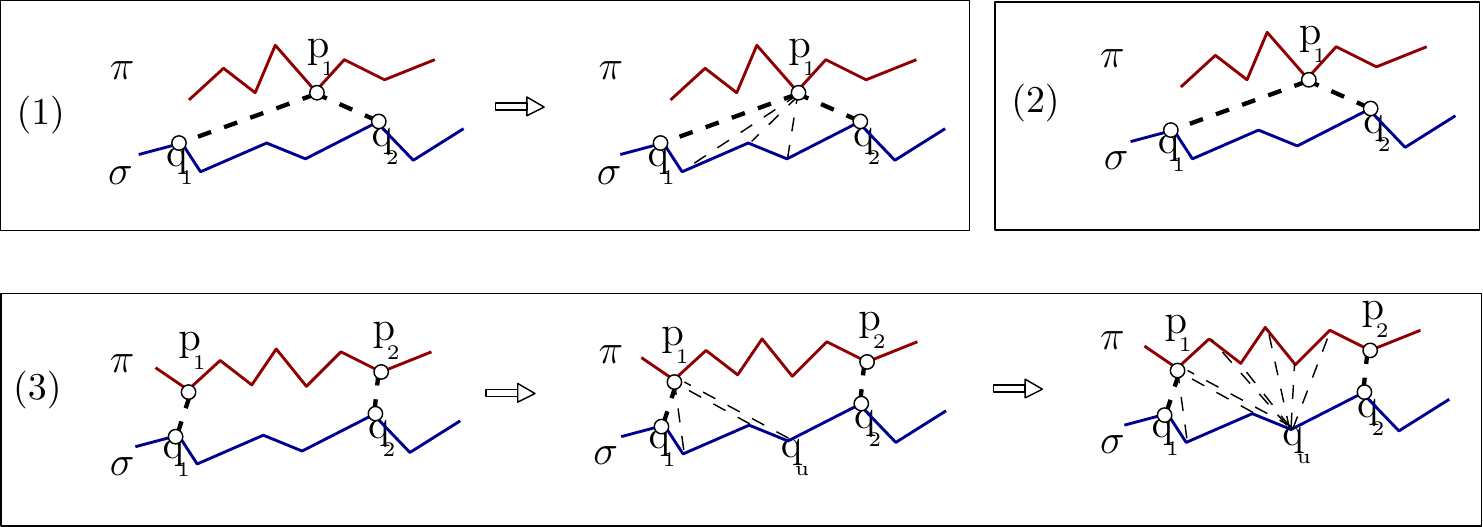}
        \caption{Small empty circles mark vertices of $\Asimp$ and
           $\Bsimp$. Three cases for two consecutive pairs from $C^*$
           between $\Asimp$ and $\Bsimp$. For case (3), we first add
           all correspondences between $p_1$ and all vertices between
           $q_1$ and $q_2$ along $\curveB$, we then add all
           correspondences between $q_u$ and vertices between $p_1$
           and $p_2$. }
        \figlab{simpapprappendix}
    \end{center}
\end{figure}
The first two are symmetric, and we simply add matches
$(p_{I_\curveA(i)}, q_k)$ (resp. $(p_k, q_{I_\curveB(j)})$) into $C$
for $k \in (I_\curveB(j), I_\curveB(j+1) )$ (resp.
$k \in (I_\curveA(i), I_\curveA(i+1) )$). For the third case, we add
all matches of the form $(p_{I_\curveA(i)}, q_k)$ for
$k \in (I_\curveB(j), I_\curveB(j+1) )$, and of the form $(p_k, q_u)$
for $k \in (I_\curveA(i), I_\curveA(i+1) )$ and
$u = I_\curveB(j+1) - 1$.  It is easy to verify that the resulting
matching $M$ is both complete and order-preserving. Furthermore, by
triangle inequality, each match $(p_i, q_j)$ added for the first two
cases satisfies $\dist(p_i, q_j) \le \DFDist(\Asimp, \Bsimp) + \mu$;
while an edge $(p_k, q_u)$ added in last case satisfies
\[
    \dist(p_k, q_u) \le \dist(p_k, p_{I_\curveA(i)}) +
    \dist(p_{I_\curveA(i)}, q_{I_\curveB(j)}) +
    \dist(q_{I_\curveB(j)}, q_u) \le \DFDist(\Asimp, \Bsimp) + 2\mu.
\]
This proves the left-hand inequality in \lemref{simpappr}.

\section{Approximating
   $\DFDist(\curveA, \curveB)$ for \backbone~curves}
\apndlab{C}

Let \algD($\curveA, \curveB, \error, \eps$) denote the $\eps$-fuzzy
decision procedure for $\DFDist(\curveA, \curveB)$ for two
\backbone~curves $\curveA$ and $\curveB$. In order to find an
$\eps$-approximation of $\DFDist(\curveA, \curveB)$, we can simply use
\thmref{fuzzy}. However, for this particular case, we can have a much
simpler binary search procedure within similar time/space complexity
that avoids the construction of WSPD.

In particular, if $\DFDist(\curveA, \curveB) < \beta$, for some
constant $\beta$, say $\beta = 1$, then we know that there exist a
pair of vertices $p^* \in \curveA$ and $q^* \in \curveB$ such that
$\dist(p^*, q^*) = \DFDist(\curveA, \curveB) < \beta$. We collect the
set $T$ of all pairs between $\curveA$ and $\curveB$ with distance
smaller than $\beta$. By similar packing argument as in
\secref{backbone} (in computing white cells), $|T| = O(n+m)$ and we
can compute $T$ in the same time/space. We then simply perform a
binary search among $T$ to locate $(p^*, q^*)$ and compute
$\DFDist(P, Q)$ exactly in $O((n+m)\log (nm))$ time. We call this
procedure \smallexactalg($P, Q, \eps, \beta$).

\begin{figure}[htb]
    \begin{center}
        \fbox{
           \begin{program}
               \>{\large{\sc{Algorithm}}}\ \ \
               \Proc{ApprFBackbone($P$, $Q$, $\eps$)} \\
               \Procbegin \\
               \>Set $\error_o = \error_n = 1, \yes = 0, \no = 0$. \\
               \> \While ( $\yes == 0$ or $\no == 0$ ) \Do \\
               \> \> \If $\error < 1 $ \Then \\
               \>    \>    \> return \smallexactalg($P, Q, \eps, 2$) \\
               \>    \> \Endif \\
               \>    \>  \If \algD($\error_n, \eps/3$) == `yes') \Then \\
               \> \> \> set $\yes = 1, \error_o = \error_n, \error_n = \error_o / (1 + \eps/3)$ \\
               \>   \> \Else \\
               \> \> \> set $\no = 1, \error_o = \error_n, \error_n = (1+\eps/3) \error_o$ \\
               \> \> \Endif \\
               \> \Endwhile \\
               \> return $\error_o$ \\
               \Endproc{}
           \end{program}
        }
    \end{center}
    \vspace{-0.5cm}
    \caption{Algorithm \apprFbackbone($P, Q, \eps$) computes an
       $\eps$-approximation of $\DFDist(P, Q)$ for two \backbone
       curves. Subroutine \smallexactalg($P, Q, \eps, \beta$) computes
       $\DFDist(P, Q)$ exactly if $\DFDist(P, Q) < \beta$. }
    % \vspace{-0.5cm}
    \figlab{fig:apprFbackbone}
\end{figure}
For the case when $\DFDist(\curveA, \curveB) \ge \beta$, we perform a
different search procedure as described in
\figref{fig:apprFbackbone}. Easy to verify that \While loop can be
called at most $O(\log_{1+\eps} (n+m)) = O((\log(n+m))/\eps)$ time, as
obviously $\DFDist(\curveA, \curveB) \le c_2(n+m)$ for \backbone
curves.  To see that the output of the algorithm is indeed an
$\eps$-approximation of $\DFDist(\curveA, \curveB)$, observe that when
the algorithm terminates, the sequence of answers from \algD() is
either a sequence of ``yes'' followed by {\bf one} ``no'', or a
sequence of ``no'' followed by {\bf one} ``yes''. Let assume that we
have the first case (the second is symmetric). Suppose the output of
the algorithm is $\bar{\error}$, then we have that
$(1-\eps/3) \bar{\error} \le \error^* = \DFDist(\curveA, \curveB)$. In
the previous iteration of the \While~loop,
$\error_n = \bar{\error} (1+\eps/3)$, as the answer then was
``yes''. Hence we have that
$\error^* \le (1+\eps/3) \error_n = (1+\eps/3)^2 \bar{\error} \le
(1+\eps) \bar{\error}$
if $\eps < 1$. This implies that $\bar{\error}$ $\eps$-approximates
$\error^*$. Hence \thmref{backboneappr} follows.

\section{Decision problem with \switchingcells~$\bwcells$ }
\apndlab{D}

The only step unexplained is how to compute $\reach[i]$ by merging
$\column[i]$ and $\reach[i-1]$ in $O(|\reach[i-1]| + |\column[i]|)$
time.  This can be achieved by a bottom-up scanning for $\reach[i-1]$
and $\column[i]$ simultaneously. More specifically, given
$\reach[i-1]$ and $\column[i]$, we can sort the endpoints of their
intervals in $O(|\reach[i-1]| + |\column[i]|)$ time using merge
sort. We process them in order and maintain the partial $\reach[i]$ at
any time. For sake of simplicity, we call an endpoint a \emph{L-point}
(resp. \emph{H-point}) of $\reach[i-1]$ or $\column[i]$ if it is the
low-endpoint (resp. higher endpoint) of some interval from
$\reach[i-1]$ or $\column[i]$. The pseudo-code is shown in
\figref{fig:mergealg}, where $\reach$ = $\reach[i-1]$,
$\column = \column[i]$, and the output is $\newreach = \reach[i]$.  It
is easy to verify the correctness of the algorithm, and the running
time is proportional to the sum of interval lists being merged.

\begin{figure}[thb]
    \begin{center}
        \fbox{
           \begin{program}
               \>{\large{\sc{Algorithm}}}\ \ \
               \Proc{\mergeCol($\reach, \column$)} \\
               \Procbegin \\
               \> \> \> Set $\potentialReach = 0$ and $\wantReach = 0$ \\
               \> \> \> Sort $H$, the set of endpoints from $\reach$ and $\column$ \\
               \> \> \> \For ( $i = 1;~i < |H|;~i ++$)~~~~\Do \\
               \> \> \> \> \If ~~(~$H[i]$ is L-point of $\reach$~)~~ \Then \\
               \> \> \> \> \> Set $\potentialReach = 1$ \\
               \> \> \> \> \> \If ~~(~$\wantReach == 1$~) \\
               \> \> \> \> \> \Then~~~~ Add $H[i]$ as L-point for $\newreach$ \\
               \> \> \> \> \Else~\If ~~(~$H[i]$ is H-point of $\reach$~)~~ \Then \\
               \> \> \> \> \> Set $\potentialReach = 0$ \\
               \> \> \> \> \Else~\If ~~(~$H[i]$ is L-point of $\column$~)~~ \Then \\
               \> \> \> \> \> Set $\wantReach = 1$ \\
               \> \> \> \> \> \If ~~(~$\potentialReach = 1$~)  \\
               \> \> \> \> \> \Then~~~~Add $H[i]$ as L-point for $\newreach$ \\
               \> \> \> \> \Else~\If ~~(~$H[i]$ is H-point of $\column$~)~~ \Then \\
               \> \> \> \> \> Set $\wantReach = 0$ \\
               \> \> \> \> \> Add $H[i]$ as H-point for $\newreach$ \\
               \> \> \> \Endfor \\
               \Endproc{}
           \end{program}
        }
    \end{center}
    \vspace{-0.5cm}
    \caption{Algorithm to compute $\reach[i]$ (i.e, $\newreach$) from
       $\reach[i-1]$ (i.e, $\reach$) and $\column[i]$ (i.e,
       $\column$). }
    \vspace{-0.5cm} \figlab{fig:mergealg}
\end{figure}

\end{document}